\documentclass[conference,letterpaper]{IEEEtran}
\pdfoutput=1
\IEEEoverridecommandlockouts
\normalsize

\usepackage{amsmath,amssymb,amsfonts} 
\allowdisplaybreaks
\usepackage{graphicx}
\usepackage{breqn}
\usepackage{cite}
\usepackage{textcomp}

\usepackage{amsthm, bm, bbm}

\DeclareMathOperator*{\argmin}{arg\,min}

\usepackage[font=small,labelfont=bf]{caption}
\usepackage{pdfpages, float}
\usepackage{mathrsfs}
\usepackage{graphicx}
\newtheorem{thm}{Theorem}
\newtheorem{col}{Corollary}
\newtheorem{lem}{Lemma}

\newtheorem{assumption}{Assumption}

\newtheorem{definition}{Definition}

\newtheorem{example}{Example}

\newtheorem{prop}{Proposition}
\usepackage{subcaption}
\definecolor{blu}{RGB}{0, 102, 204}
\definecolor{purp}{RGB}{128,0,128}
\definecolor{rd}{RGB}{255,69,0}
\definecolor{org}{RGB}{255, 95, 31}
\definecolor{cyn}{RGB}{0, 200, 200}

\usepackage[linesnumbered,ruled,vlined]{algorithm2e}
\usepackage{setspace}

\SetCommentSty{mycommfont}

\date{}
\begin{document}

\title{
	Sequential Interval Passing for Compressed Sensing 
}
\author{%
  \IEEEauthorblockN{Salman Habib	\thanks{This work was supported in part by the NSF and Office of the Under Secretary of Defense (OUSD) – Research and Engineering, Grant ITE2515378, as part of the NSF Convergence Accelerator Track G: Securely Operating Through 5G Infrastructure Program.} and Taejoon Kim}
  \IEEEauthorblockA{School of Electrical, Computer and Energy Engineering \\
                    Arizona State University\\
                    Email: \{shabib13, taejoonkim\}@asu.edu}
  \and
  \IEEEauthorblockN{R\'{e}mi A. Chou}
  \IEEEauthorblockA{Department of Computer Science and Engineering\\ 
                    University of Texas at Arlington\\
                    Email: remi.chou@uta.edu}
}

\maketitle
\begin{abstract} 
The reconstruction of sparse signals from a limited set of measurements poses a significant challenge as it necessitates a solution to an underdetermined system of linear equations. Compressed sensing (CS) deals with sparse signal reconstruction using techniques such as linear programming (LP) and iterative message passing schemes. The interval passing algorithm (IPA) is an attractive CS approach due to its low complexity when compared to LP. In this paper, we propose a sequential IPA that is inspired by sequential belief propagation decoding of low-density-parity-check (LDPC) codes used for forward error correction in channel coding. In the sequential setting, each check node (CN) in the Tanner graph of an LDPC measurement matrix is scheduled one at a time in every iteration, as opposed to the standard ``flooding'' interval passing approach in which all CNs are scheduled at once per iteration. The sequential scheme offers a significantly lower message passing complexity compared to flooding IPA on average, and for some measurement matrix and signal sparsity, a complexity reduction of $36\%$ is achieved. We show both analytically and numerically that the reconstruction accuracy of the IPA is not compromised by adopting our sequential scheduling approach. 
\end{abstract}

\section{Introduction}
\label{sec:Intro}

Consider a signal $\mathbf{x}\in \mathbb{R}^n$, and a signal $\mathbf{y}\in \mathbb{R}^m$ containing $m$ measurements of $\mathbf{x}$ subject to the condition $\mathbf{y}=\mathbf{A}\mathbf{x}$, where $\mathbf{A}\in \mathbb{R}^{m\times n}$ is a measurement matrix. The entries $x(1),\ldots,x(n)$ of $\mathbf{x}$ can be determined from $\mathbf{y}$ by solving a system of $m$ linear equations, each containing $n$ unknowns.  However, if the signal $\mathbf{y}$ is compressed, \emph{i.e.,} if $m<n$, then the system is underdetermined and can have infinitely many solutions. Compressed sensing (CS) is a technique used to correctly reconstruct an unknown signal $\mathbf{x}$ from its compressed version $\mathbf{y}$, which inevitably faces the challenge of solving an undetermined system of linear equations. It has been shown in~\cite{Candes} that CS requires a much smaller number of samples than required by the Shannon-Nyquist sampling rule to correctly reconstruct $\mathbf{x}$, provided that the signal is $k$-sparse, \emph{i.e.,}$k\ll n$ is the  number of non-zero elements in $\mathbf{x}$. One approach focuses on finding a reconstructed version $\mathbf{\hat{x}}$ of $\mathbf{x}$ which has the smallest $\ell_0$ norm \cite{Don06}. However, this approach can be computationally infeasible for practical applications due to its NP-hard nature. There exists a method based on linear programming (LP) for finding $\mathbf{\hat{x}}$ \cite{CDS98}, which is also computationally intensive, especially for large measurement matrices.

To alleviate the complexity problem, iterative message passing techniques such as the interval passing algorithm (IPA) has been developed for reconstructing sparse and non-negative $\mathbf{x}$ \cite{Chandar}. 
IPA operates on the Tanner graph of low-density-parity-check (LDPC) measurement matrices \cite{Tan81}. The Tanner graph of a $(\gamma,\rho)$ regular LDPC  measurement matrix consists of check nodes (CNs) with degree $\rho$ and variable nodes (VNs) with degree $\gamma$, and edges connecting them. There exists other iterative compressed sensing approaches such as those in \cite{Sarvotham,Hoa}, and the approximate message passing (AMP) algorithms \cite{DonMalMon09,Bayati,Chunli,Nikolajs,Kim15,10206981}. The construction of quasi-cyclic spatially coupled LDPC measurement matrices for the IPA was discussed in \cite{mycs}. Although the IPA is less accurate compared to LP, its computational complexity is significantly lower than LP. In~\cite{Shantharam}, a low complexity IPA is proposed for binary signals by converting the maximization operation of the IPA into a logical OR operation. In \cite{Ravanmehr}, the conditions based on which the IPA can overcome stopping set errors have been identified, including upper bounds on the size of stopping sets for some LDPC measurement matrices.

In this paper, we propose a sequential IPA in which all the CNs in the Tanner graph of an LDPC measurement matrix are randomly scheduled in each iteration. A CN scheduling operation refers to updating all the neighboring VNs of this CN, as well as updating all the CN neighbors of these VNs. In contrast, the existing IPA is based on a flooding schedule in which all CNs are updated first, and then all VNs are updated, repeating in every iteration \cite{Ravanmehr}. Our proposed sequential IPA is mainly inspired by the computational advantage of the sequential BP decoder. We show empirically that similar to the channel coding case, sequential IPA requires fewer number of iterations for convergence, on average, than the flooding IPA for a given measurement matrix and signal sparsity, resulting in a fewer number of CN to VN message updates overall. Our sequential interval passing scheme is applicable to both binary and non-binary CS. 
The main contributions of this paper are as follows. 
\begin{itemize}
	\item We incorporate the scheduling time of a CN in the proposed sequential IPA to distinguish between CNs which are scheduled prior to scheduling the current node in a given iteration. In comparison, the iterative CS methods discussed in \cite{Chandar,Krishnan,Ravanmehr} lack the concept of scheduling time as all CNs are scheduled simultaneously per iteration. 
	\item We analytically show that the sequential IPA is superior to the flooding IPA both in terms of computational efficiency and reconstruction accuracy. The fact that a failure of the sequential IPA guarantees failure of the flooding IPA is also shown. Additionally, we prove that for a $(2,3)$ LDPC measurement matrix, the frame error rate (FER) of the sequential IPA is upper bounded by the FER of the flooding IPA. 
\end{itemize}

\section{Preliminaries}
CS is the process of reconstructing an $n$-dimensional sparse signal $\mathbf{x}$ from an $m$-dimensional measurement $\mathbf{y}=\mathbf{A}\mathbf{x}$, where $m<n$ and $\mathbf{A}\in \mathbb{R}^{m\times n}$ is the measurement matrix. A CS algorithm succeeds at correctly retrieving the unknown signal $\mathbf{x}$ only if it is sparse, \emph{i.e.,} $\mathbf{x}$ must contain $k\ll n$ non-zero elements. A straightforward approach to solving $\mathbf{x}$ from $\mathbf{y}$ is expressed as an optimization problem \cite{Don06}
\begin{equation}
	\argmin \lVert\mathbf{x}\lVert_0~\mathrm{s.t.}~ \mathbf{y}=\mathbf{A}\mathbf{x},
\end{equation}
which is known as $\ell_0$-norm minimization. This approach is NP-hard as the search space consists of $n\choose k$ possible candidates. A less computationally intensive approach, expressed as
\begin{equation}
\argmin \lVert \mathbf{x}\lVert_1~\mathrm{s.t.}~ \mathbf{y}=\mathbf{A}\mathbf{x},
\end{equation}
minimizes the $\ell_1$-norm of the signal and can be solved via LP which has polynomial complexity~\cite{CDS98}. 
It has been shown in \cite{Ravanmehr} that the reconstruction complexity of CS can be further reduced by employing iterative schemes such as the IPA that operates on the Tanner graph of an LDPC measurement matrix. The IPA has complexity even lower than LP, albeit at a lower reconstruction accuracy. 

An LDPC measurement matrix $\mathbf{A}\in \mathbb{R}_{\geq 0}^{m\times n}$, $m< n$, can be represented as a \emph{Tanner graph} $\mathcal{G}_{\mathbf{A}}=(\mathcal{V}\cup \mathcal{C},\mathcal{E})$, where $\mathcal{V}=\{v_1,\ldots,v_n\}$ is a set of VNs corresponding to the columns of $\mathbf{A}$, $\mathcal{C}=\{c_1,\ldots,c_m\}$ is a set of CNs corresponding to the rows of $\mathbf{A}$, and $\mathcal{E}$ is a set of edges connecting $\mathcal{V}$ to $\mathcal{C}$ that correspond to the non-zero elements in $\mathbf{A}$ \cite{Tan81}. Let $A_{c_j,v_i}$ denote an entry in row $c_j$ and column $v_i$ of $\mathbf{A}$. Since LDPC measurement matrices are sparse, $\mathbf{A}$ has a much smaller number of non-zero entries compared to its dimension. 
Let $\mu_{v_i \rightarrow c_j}^{(\ell)}$ (resp., $M_{v_i \rightarrow c_j}^{(\ell)}$) denote a message lower (resp., upper) bound propagated by VN $v_i$ to CN $c_j$ during iteration $\ell$. Likewise, $\mu_{c_j \rightarrow v_i}^{(\ell)}$ (resp., $M_{c_j \rightarrow v_i}^{(\ell)}$) refers to a message lower (resp., upper) bound propagated by CN $c_j$ to VN $v_i$ during iteration $\ell$ (see Fig. \ref{fig:ga}). Also let $\mathcal{N}(v_i)$ be the set of neighboring CNs of VN $v_i$, and $\mathcal{N}(c_j)$ be the set of neighboring VNs of CN $c_j$. More specifically, the lower bound of a VN to CN message in iteration $\ell$ is given by
\begin{align}
\mu_{v_i \rightarrow c_j}^{(\ell)}= 
\max_{c'\in \mathcal{N}(v_i)} \mu_{c' \rightarrow v_i}^{(\ell-1)}\times A_{c_j,v_i},
\label{eq:mu_vc_fld}
\end{align}
whereas a CN to VN message lower bound is computed as 
\begin{equation}
	\mu_{c_j \rightarrow v_i}^{(\ell)}=
	\begin{cases}
	0, \mathrm{if}~\ell=0, \\
	\frac{1}{A_{c_j,v_i}} \left( y(c_j)-\sum_{v'\in \mathcal{N}(c_j)\setminus v_i}M_{v' \rightarrow c_j}^{(\ell)}\right), \mathrm{if}~\ell>0.
	\end{cases} 
\label{eq:mu_cv_fld}
\end{equation}

\begin{figure}[h]
  \centering
  \includegraphics[scale=0.6]{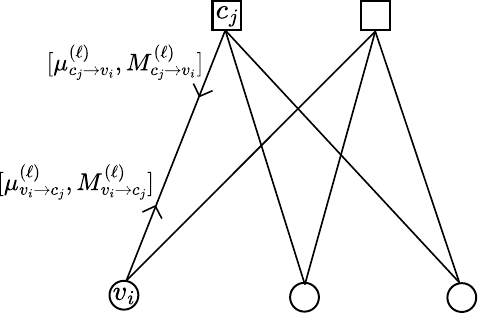}
  \caption{Example of $\mathcal{G}_{\mathbf{A}}$ corresponding to a $(2,3)$ LDPC measurement matrix with $m=2$ and $n=3$. Circles represent VNs and squares represent CNs. The intervals passed by VNs $v_i$ and $c_j$ are also shown.}
\label{fig:ga}
\end{figure}

In contrast, the corresponding message upper bounds are expressed as
\begin{align}
M_{v_i \rightarrow c_j}^{(\ell)}= 
\min_{c'\in \mathcal{N}(v_i)} M_{c' \rightarrow v_i}^{(\ell-1)}\times A_{c_j,v_i}, 
\label{eq:M_vc_fld}
\end{align}
and
\begin{align}
M_{c_j \rightarrow v_i}^{(\ell)}=
	\begin{cases}
	y(c_j)/A_{c_j,v_i}, \mathrm{if}~\ell=0, \\
	\frac{1}{A_{c_j,v_i}} \left( y(c_j)-\sum_{v'\in \mathcal{N}(c_j)\setminus v_i}\mu_{v' \rightarrow c_j}^{(\ell)}\right), \mathrm{if}~\ell>0.
	\end{cases} 
\label{eq:M_cv_fld}
\end{align}  
For a given signal $\mathbf{x}$, a maximum of $\ell_\max$ iterations are performed. In the remainder of the paper let $[n]\triangleq \{1,\ldots,n\}$, $n\in \mathbb{N}$. The standard (flooding) IPA, abbreviated as $\mathrm{FIPA}(\mathbf{y}, \mathbf{A})$, or FIPA, is shown in Algorithm 1 of \cite{Ravanmehr}. 
The message $\mu_{v_i \rightarrow \mathcal{N}({v_i})}^{(\ell)}$ represents a lower bound propagated by VN $v_i$ to one of the neighboring CNs in $\mathcal{N}({v_i})$, and the upper bound propagated by this CN to $v_i$ is expressed as $M_{\mathcal{N}({v_i})\rightarrow v_i}^{(\ell)}$.

\section{The Sequential Interval Passing Algorithm}  
In our proposed sequential IPA, abbreviated as $\mathrm{SIPA}(\mathbf{y}, \mathbf{A})$, or SIPA, all the CNs in $\mathcal{G}_{\mathbf{A}}$ are serially scheduled in each iteration. Since a VN is connected to more than one CN, the VN can be updated more than once per iteration, unlike flooding in which a VN is updated only once per iteration. At a given time, a CN is scheduled randomly from $\mathcal{C}\setminus \mathcal{C}_e$, where each CN in $\mathcal{C}_e$ has been scheduled at an earlier time. It is worthwhile to note that the order in which the CNs are scheduled can  be learned with the aid of reinforcement learning (RL), and it has been shown in \cite{reldec} that an RL-based CN scheduling order outperforms a randomly generated one when soft valued sequential BP decoding of LDPC codes is considered. However, in this work, experiments reveal that a random CN scheduling order for SIPA performs just as well as an RL-based CN scheduling order. This is because in the case of SIPA, the long term rewards of the CNs, earned by the RL algorithm over many signals and iterations, are almost identical unlike channel coding. Hence, in the remainder of the paper we only consider random scheduling of the CNs for sequential interval passing.

Let $\Gamma_v=\{t_1,\ldots,t_{\gamma}| i\in [\gamma],t_j \in [m]\}$ represent the scheduling order of all the CNs connected to a VN $v\in \mathcal{V}$ with degree $\gamma$, and $t_j$ is the scheduling time of  the $j$-th neighboring CN of VN $v$. Suppose that $c_j$ is scheduled at current time $t_j$, $c_e$ is scheduled at an earlier time $t_e$, and $c_l$ is scheduled at a later time $t_l$, \emph{i.e.,} $t_e<t_j<t_l$, and $e,j,l\in [m]$. Then, a VN to CN message lower bound in case of SIPA is expressed as
\begin{align}
	\mu_{v_i \rightarrow c_j}^{(\ell,t_j)}&=\max\bigg(\max_{c_e\in \mathcal{N}(v_i)\setminus \{c_j,\mathcal{C}_l(v_i)\}}\mu_{c_e \rightarrow v_i}^{(\ell,t_e)},\mu_{c_j \rightarrow v_i}^{(\ell-1,t_j)}, \nonumber \\
	&~~~~~ \max_{c_l\in \mathcal{N}(v_i)\setminus\{c_j,\mathcal{C}_e(v_i)\}}\mu_{c_l \rightarrow v_i}^{(\ell-1,t_l)}\bigg)\times A_{c_j,v_i}. 
\label{eq:mu_vc_seq}
\end{align}
where $\mathcal{C}_l(v_i)$ is a set of all the CNs in $\mathcal{N}(v_i)$ that will be scheduled at a later time $t_l>t_j$, and $\mathcal{C}_{e}(v_i)$ is a set of all the CNs in $\mathcal{N}(v_i)$ that was scheduled at an earlier time $t_e<t_j$. Note that $\mathcal{C}_e(v_i)=\emptyset$ if $t_j=1$, and $\mathcal{C}_l(v_i)=\emptyset$ if $t_j=m$. 

In (\ref{eq:mu_vc_seq}), a CN to VN message lower bound is given by
\begin{equation}
	\mu_{c_j \rightarrow v_i}^{(\ell,t_j)}=
	\begin{cases}
	&0, \mathrm{if}~\ell=0,t_j\leq m, \\
	&\mu_{c_j \rightarrow v_i}^{(\ell-1,t)}, \mathrm{if}~\ell> 1, ~t_j=0,~t\geq 0, \\
	&\frac{1}{A_{c_j,v_i}} \left( y(c_j)-\sum_{v'\in \mathcal{N}(c_j)\setminus v_i}M_{v' \rightarrow c_j}^{(\ell,t_{v'})}\right), \\
	&~~~~~~~~~~~~~~~~~~ \mathrm{if}~\ell>0, ~t_j> 0,t_{v'}\leq t_j, 
	\end{cases} 
\label{eq:mu_cv_seq}
\end{equation} 
\noindent where $t_j$ is time at which CN $c_j$ is scheduled, and $t_{v'}$ is the scheduling time of a CN $c\in \mathcal{N}(v')$. Note that $t_j=0$ indicates a CN $c_j$ has not been scheduled, which is possible if partial scheduling is done, \emph{i.e.,} only a subset of CNs are scheduled every iteration. However, this strategy is not considered in this work as we schedule all CNs in every iteration. 

Likewise, for the upper bounds,
\begin{align}
	M_{v_i \rightarrow c_j}^{(\ell,t_j)}&=\min\bigg(\min_{c_e\in \mathcal{N}(v_i)\setminus \{c_j,\mathcal{C}_l(v_i)\}}M_{c_e \rightarrow v_i}^{(\ell,t_e)},M_{c_j \rightarrow v_i}^{(\ell-1,t_j)}, \nonumber \\
	&~~~~~ \min_{c_l\in \mathcal{N}(v_i)\setminus\{c_j,\mathcal{C}_e(v_i)\}}M_{c_l \rightarrow v_i}^{(\ell-1,t_l)}\bigg)\times A_{c_j,v_i}, 
\label{eq:M_vc_seq}
\end{align}
where
\begin{equation}
M_{c_j \rightarrow v_i}^{(\ell,t_j)}=
	\begin{cases}
	&y(c_j)/A_{c_j,v_i}, \mathrm{if}~\ell=0, t_j\leq m, \\ 
	&M_{c_j \rightarrow v_i}^{(\ell-1,t)}, \mathrm{if}~\ell> 1, ~t_j=0,~t\geq 0, \\
	&\frac{1}{A_{c_j,v_i}} \left( y(c_j)-\sum_{v'\in \mathcal{N}(c_j)\setminus v_i}\mu_{v' \rightarrow c_j}^{(\ell,t_{v'})}\right),  \\
	&~~~~~~~~~~~~~~~~~~  \mathrm{if}~\ell> 0, ~t_j> 0,t_{v'}\leq t_j. 
	\end{cases} 
\label{eq:M_cv_seq}
\end{equation}  

{\linespread{0.2}\selectfont  
\begin{algorithm}[tb]
\caption{The sequential IPA: $\mathrm{SIPA}(\mathbf{y}, \mathbf{A})$}\label{alg:sipa}
\SetAlgoLined
\DontPrintSemicolon
\SetKwInOut{Input}{Input}
\SetKwInOut{Output}{Output}
\textbf{Input: }$\mathbf{y}$ and $\mathbf{A}$ such that $\mathbf{y}=\mathbf{A}\mathbf{x}$\;
\textbf{Output: }reconstructed signal $\mathbf{\hat{x}}=[\hat{x}(1),\ldots,\hat{x}(n)]$\;
\textbf{Initialization:} for all $c_j\in \mathcal{C}$, $v_i \in \mathcal{N}(c_j)$, $\mu_{c_j \rightarrow v_i}^{(0,0)}=0$ and $M_{c_j \rightarrow v_i}^{(0,0)}=y(c_j)/A_{c_j,v_i}$, $p_{v_i}^{(\ell)}\leftarrow 0 ~\forall v_i,\ell$, $\ell\leftarrow 1$\;

\While{iteration $\ell< \ell_{\max}$} { 
	\For{each CN $c_j\in \mathcal{C}$ scheduled at time $t_j$}{ 
	\For{each VN $v_i\in \mathcal{N}(c_j)$}{
		\For{each CN $c'\in \mathcal{N}(v_i)$}{
			\If{$p_{v_i}^{(\ell)}=0 \vee \mu_{v_i \rightarrow c'}^{(\ell-1,t_{v_i})}< M_{v_i \rightarrow c'}^{(\ell-1,t_{v_i})}$} { 
			compute and propagate $\mu_{v_i \rightarrow c'}^{(\ell,t_j)}$ according to (\ref{eq:mu_vc_seq})\; 
			compute and propagate $M_{v_i \rightarrow c'}^{(\ell,t_j)}$ according to (\ref{eq:M_vc_seq})\; 		
			$p_{v_i}^{(\ell)}\leftarrow p_{v_i}^{(\ell)}+1$\;
			}
		}
	}
		\For{each VN $v_i\in \mathcal{N}(c_j)$}{
			compute and propagate $\mu_{c_j \rightarrow v_i}^{(\ell,t_j)}$ according to (\ref{eq:mu_cv_seq})\; 
			\If{$\mu_{c_j \rightarrow v_i}^{(\ell,t_j)}< 0$} {
				$\mu_{c_j \rightarrow v_i}^{(\ell,t_j)}=0$\;
			}
			compute and propagate $M_{c_j \rightarrow v_i}^{(\ell,t_j)}$ according to (\ref{eq:M_cv_seq})\; 
		}
	}
	\tcp{make decision}
	\For{each VN $v_i\in \mathcal{V}$}{
		\If{$\bigg(\mu_{v_i \rightarrow \mathcal{N}({v_i})}^{(\ell,t_{\mathcal{N}({v_i})})}=M_{\mathcal{N}({v_i})\rightarrow v_i}^{(\ell,t_{\mathcal{N}({v_i})})} ~\vee~ \ell=\ell_\max\bigg) $ } {\; 
			$\hat{x}(v_i)=\mu_{v_i \rightarrow \mathcal{N}({v_i})}^{(\ell,t_{\mathcal{N}({v_i})})}$\;	
		}
	}
	$\ell\leftarrow \ell+1$\;
}
\end{algorithm}
}

Our proposed sequential IPA is shown in Algorithm \ref{alg:sipa} of this paper. Once a CN $c_j$ is scheduled at time $t_j$, Steps 6-21 of this algorithm are performed in its sub-graph. To prevent redundant VN updates, we utilize the variable $p_{v_i}^{(\ell)}$ in Step 11 to record the number of times VN $v_i$ is updated during iteration $\ell$.  A VN is updated more than once per iteration only if the condition in Step 8 is true. 

The intuitive reason for the improved performance of SIPA over the FIPA is that the former utilizes messages propagated by the earlier scheduled CNs, in the same iteration, for updating some or all of the VNs connected to the currently scheduled CN. On the other hand, in flooding, all VNs are updated using CN to VN messages computed in the previous iteration.

\section{Complexity Reduction Via Sequential CN Scheduling}
In this section we discuss conditions under which SIPA requires fewer number of iterations than FIPA for correctly reconstructing a signal, and then provide illustrative examples for clarification. Note that in Step 24 of Algorithm \ref{alg:sipa}, if $\mu_{v_i \rightarrow c_j}^{(1,t_j)}$ matches $M_{c_j \rightarrow v_i}^{(1,t_j)}$ and they both equal $x(v_i)$, then SIPA will correctly output $\hat{x}(v_i)$. If $\mu_{v_i \rightarrow c_j}^{(1,t_j)}$ (resp., $M_{c_j \rightarrow v_i}^{(1,t_j)}$) is larger (resp., smaller) than its initial value, which is zero (resp., $y(c_j)/A_{c_j,v_i}$), then the chance that a lower bound propagated by VN $v_i$ to CN $c_j$ matches the upper bound propagated to this node by $c_j$ increases, and so does thus the chance of correctly reconstructing a non-zero $x(v_i)$. The following lemma establishes a condition under which SIPA outperforms its flooding counterpart.
\begin{lem} 
There exists a $x(v_i)\in \mathbb{R}_{\geq 0}$ and $\mathbf{A} \in \mathbb{R}_{\geq 0}^{m\times n}$ for which SIPA propagates $\mu_{v_i \rightarrow c_j}^{(1,t_j)}>0$ and $M_{c_j \rightarrow v_i}^{(1,t_j)}\leq y(c_j)/A_{c_j,v_i}$. 
\label{lem:iter1}
\end{lem}
\begin{proof}
If $x(v_i)>0$, then for a CN $c_j\in \mathcal{N}(v_i)$, we have $y(c_j)>0$. Consider two message passing events
\begin{align}
	E_1&\triangleq \bigcup_{c_j \in \mathcal{N}(v_i)}\bigg\{\{y(c_j)>0\} \bigcap \bigg\{\sum_{v'\in \mathcal{N}(c_j)\setminus v_i}M_{v' \rightarrow c_j}^{(1,t_{v'})}< \nonumber \\
	&~~~~~ y(c_j)\bigg\}\bigg\}
\end{align}
and
\begin{align}
	E_2& \triangleq \bigcup_{c_j \in \mathcal{N}(v_i)}\bigg\{\{y(c_j)>0\} \bigcap \bigg\{\sum_{v'\in \mathcal{N}(c_j)\setminus v_i}\mu_{v' \rightarrow c_j}^{(1,t_{v'})}\leq \nonumber \\
	&~~~~~ y(c_j)/A_{c_j,v_i} \bigg\}\bigg\},
\end{align}
where $\cup$ and $\cap$ denote arbitrary union and intersection operations, respectively, and $t_{v'}\leq t_j$. If $E_1$ is true, then according to (\ref{eq:mu_cv_seq}) $\mu_{c_j \rightarrow v_i}^{(1,t_j)}>0$, and hence from (\ref{eq:mu_vc_seq}) we get $\mu_{v_i \rightarrow c_j}^{(1,t_j)}>0$. Additionally, if $E_2$ is true, then from (\ref{eq:M_cv_seq}) we get $M_{c_j \rightarrow v_i}^{(1,t_j)}\leq y(c_j)/A_{c_j,v_i}$. This completes the proof.
\end{proof}
In the following proposition, we will invoke events $E_1$ and $E_2$ from the proof of Lemma~\ref{lem:iter1} to state a condition for correct recovery of a non-zero binary signal  $x(v_i)$ using $\mathrm{SIPA}(\mathbf{y},\mathbf{A})$ when $\ell=1$.
\begin{prop}
For a given $\mathbf{x} \in \mathbb{F}_2^n$ and $\mathbf{A} \in \mathbb{F}_2^{m\times n}$ SIPA correctly recovers $x(v_i)$ if the event
\begin{equation}
	E_1 \cap E_2 \bigcap \bigg\{\bigcup_{c_j\in \mathcal{N}(v_i)} \bigg\{\mu_{v_i \rightarrow c_j}^{(1,t_j)}=1 \wedge M_{c_j \rightarrow v_i}^{(1,t_j)}=1\bigg\}\bigg\}
	\label{eq:propeq1}
\end{equation}
occurs. However, FIPA  propagates $\mu_{v_i \rightarrow c_j}^{(1)}=0$ and $M_{c_j \rightarrow v_i}^{(1)}=y(c_j)$ irrespective of $x(v_i)$ and $\mathbf{A}$, causing IPA failure when $\ell=1$.
\label{prop:first_res}
\end{prop}
\noindent The proposition can be verified by observing Step 3 of Algorithm 1 in \cite{Ravanmehr} and Steps 23-28 of Algorithm \ref{alg:sipa} in this paper. We can see that if $\mu_{v_i \rightarrow c_j}^{(1)}=0$ and $M_{c_j \rightarrow v_i}^{(1)}=y(c_j)$, then the condition in Step 3 of Algorithm 1 in \cite{Ravanmehr} is not satisfied, causing failure of FIPA. But for the sequential case, the signal $x(v_i)$ is correctly reconstructed. Note that in (\ref{eq:propeq1}) the condition of Lemma~\ref{lem:iter1} is true as $\mu_{v_i \rightarrow c_j}^{(1,t_j)}>0$ and $M_{c_j \rightarrow v_i}^{(1,t_j)}\leq y(c_j)$.

In the following, we show that for some $x(v_i)$, $\mathrm{SIPA}(\mathbf{y},\mathbf{A})$ more accurately reconstructs $x(v_i)$ than $\mathrm{FIPA}(\mathbf{y},\mathbf{A})$ for $\ell>1$, and we show in Corollary \ref{col:seq<fld} that failure of SIPA implies failure of FIPA. This is true because in the former case, the lower bound $\mu_{v_i \rightarrow c_j}^{(\ell,t_j)}$ propagated by VN $v_i$ is larger than $\mu_{v_i \rightarrow c_j}^{(\ell)}$ propagated in flooding, and it also receives an upper bound $M_{c_j \rightarrow v_i}^{(\ell,t_j)}$ no larger than the corresponding upper bound $M_{c_j \rightarrow v_i}^{(\ell)}$ in flooding. A larger lower bound propagated by a VN is more likely to be non-zero if $x(v_i)\neq 0$, and a smaller upper bound propagated to the same VN is more likely to match its lower bound. 
\begin{lem} 
There exists an $x(v_i)\in \mathbb{R}_{\geq 0}$ and $\mathbf{A} \in \mathbb{R}_{\geq 0}^{m\times n}$ for which $\mu_{v_i \rightarrow c_j}^{(\ell,t_j)}>\mu_{v_i \rightarrow c_j}^{(\ell)}$ and $M_{c_j \rightarrow v_i}^{(\ell,t_j)}\leq M_{c_j \rightarrow v_i}^{(\ell)}$, for any $\ell>1$.
\label{lem:l>1}
\end{lem}
\noindent If the events
\begin{align*}
E_3&\triangleq \bigcup_{c\in \mathcal{N}(v_i)} \bigg\{\mu_{c \rightarrow v_i}^{(\ell-1,t)}>\mu_{c \rightarrow v_i}^{(\ell-1)} \bigg\}, \\
E_4&\triangleq E_3 \bigcap_{c'\in \mathcal{N}(v_i)\setminus c}\bigg\{\mu_{c' \rightarrow v_i}^{(\ell-1,t)}=\mu_{c' \rightarrow v_i}^{(\ell-1)} \bigg\}, \\
	E_5&\triangleq \bigcap_{v\in \mathcal{N}(c_j)} \bigg\{M_{c_j \rightarrow v}^{(\ell,t_j)}\leq M_{c_j \rightarrow v}^{(\ell)}\bigg\}
\end{align*}
are true, we can state a condition for correct recovery of a non-zero binary signal  $x(v_i)$ using $\mathrm{SIPA}(\mathbf{y},\mathbf{A})$ when $\ell>1$.
\begin{prop}
For a given $\mathbf{x} \in \mathbb{F}_2^n$ and $\mathbf{A} \in \mathbb{F}_2^{m\times n}$, SIPA correctly recovers a non-zero $x(v_i)$ if the event
\begin{equation}
	E_5 \bigcap \bigg\{\bigcup_{c_j\in \mathcal{N}(v_i)} \bigg\{\mu_{v_i \rightarrow c_j}^{(\ell,t_j)}=1 \wedge M_{c_j \rightarrow v_i}^{(\ell,t_j)}=1\bigg\}\bigg\}
	\label{eq:propeq2}
\end{equation}
is true, but FIPA cannot recover $x(v_i)$ because $E_4$ implies $\mu_{v_i \rightarrow c_j}^{(\ell)}<\mu_{v_i \rightarrow c_j}^{(\ell,t_j)}$. 
\label{prop: l>1}
\end{prop}
\noindent This proposition can be verified in a similar fashion to Proposition~\ref{prop:first_res}. Note that if Proposition~\ref{prop: l>1} is true for a given $\mathbf{x} \in \mathbb{F}_2^n$ and $\mathbf{A} \in \mathbb{F}_2^{m\times n}$, then FIPA can recover $x(v_i)$ if it performs a larger number of iterations than SIPA. Moreover, in (\ref{eq:propeq2}), Lemma~\ref{lem:l>1} holds since $\mu_{v_i \rightarrow c_j}^{(1,t_j)}>0$ and $M_{c_j \rightarrow v_i}^{(1,t_j)}\leq y(c_j)$.

\begin{thm} 
For a given $\mathbf{x} \in \mathbb{R}_{\geq 0}^n$ and $\mathbf{A} \in \mathbb{R}_{\geq 0}^{m\times n}$, if SIPA propagates $\mu_{v_i \rightarrow c_j}^{(\ell,t_j)}= 0$, then FIPA will propagate $\mu_{v_i \rightarrow c_j}^{(\ell)}=0$.
\label{thm:fld<seq}
\end{thm}

\begin{proof}
The IPA ensures that $M_{v \rightarrow c}^{(l)}\leq M_{v \rightarrow c}^{(l-1)}$ and $M_{v \rightarrow c}^{(l,t)}\leq M_{v \rightarrow c}^{(l-1,t)}$. Since a VN can be updated more than once in the same iteration during sequential scheduling, a VN to CN message is updated more frequently than flooding, provided the number of iterations for both schemes is identical. As a result, we have $M_{v \rightarrow c}^{(l,t)}\leq M_{v \rightarrow c}^{(l)}$, leading to $\mu_{c \rightarrow v}^{(\ell,t)}\geq \mu_{c \rightarrow v}^{(\ell)}$ based on (\ref{eq:mu_cv_fld}) and (\ref{eq:mu_cv_seq}). If SIPA propagates $\mu_{v_i \rightarrow c_j}^{(\ell,t_j)}=0$ for a given $\mathbf{x}$, it would imply according to (\ref{eq:mu_vc_seq}) that $\mu_{c_j \rightarrow v_i}^{(\ell-1,t_j)}= 0$ for all $c_j \in \mathcal{N}(v_i)$. Consequently, in case of flooding, it must be true that $\mu_{c_j \rightarrow v_i}^{(\ell-1)}= 0$ for all $c_j \in \mathcal{N}(v_i)$, and hence according to (\ref{eq:mu_vc_fld}), $\mu_{v_i \rightarrow c_j}^{(\ell)}=0$ for that $\mathbf{x}$.
\end{proof}

\begin{col} 
Theorem~\ref{thm:fld<seq} implies that if SIPA cannot recover a non-zero signal, $x(v_i)$, of $\mathbf{x}$ during iteration $\ell$, then FIPA also cannot recover it. 
\label{col:seq<fld}
\end{col}

\section{FER Comparison Between FIPA and SIPA}

In this section, we show that by allowing a reduction in complexity, SIPA does not incur a worse FER than its flooding counterpart, where FER is measured as $\Pr(\mathbf{x}\neq \mathbf{\hat{x}})$. Let $\mathrm{supp}(\mathbf{x})$ denote the support, \emph{i.e.,} the set of non-zero entries, of $\mathbf{x}$. For a measurement matrix with Tanner graph $\mathcal{G}_{\mathbf{A}}$ and signal weight $K=|\mathrm{supp}(\mathbf{x})|$, let $\mathrm{FER}(\mathcal{G}_{\mathbf{A}},K,\ell)$ and $\mathrm{FER}(\mathcal{G}_{\mathbf{A}},K,\ell,\Gamma_v)$ be the FER of FIPA and SIPA, respectively, where $\Gamma_v$ represents the CN scheduling order of any VN $v\in \mathcal{V}$. Suppose $\mathcal{G}_{\mathbf{A}}$ corresponds to a $(2,3)$ LDPC measurement matrix. Then, we get the following result.
\begin{lem}
For a $(2,3)$ LDPC measurement matrix, $\mathrm{FER}(\mathcal{G}_{\mathbf{A}},K,\ell,\Gamma_v)\leq \mathrm{FER}(\mathcal{G}_{\mathbf{A}},K,2)$, for any $\ell>1$. 
\label{lem:bound}
\end{lem}
\begin{thm}
For a $(2,3)$ LDPC measurement matrix, $\mathrm{FER}(\mathcal{G}_{\mathbf{A}},K,\ell,\Gamma_v)\leq \mathrm{FER}(\mathcal{G}_{\mathbf{A}},K,\ell)$, for any $\ell>0$.
\label{thm:FER}
\end{thm}
Theorem \ref{thm:FER} shows that by allowing a reduction in complexity, SIPA does not incur a worse FER than its flooding counterpart.
 
\vspace{-0.5cm}
\section{Simulation Results}
We compare FIPA to SIPA using a $(3,7)$ LDPC measurement matrix $\mathbf{A}_1\in \mathbb{F}_2^{400\times 700}$ which is randomly generated. For interval passing we choose $\ell_\max=50$. The probability of correct reconstruction (PCR) is given by $1-\Pr(\mathbf{x}\neq \mathbf{\hat{x}})$, and used as a measure of the reconstruction performance of a measurement matrix for different values of signal sparsity $k/n$, where $k$ is the number of non-zero entries in $\mathbf{x}\in \mathbb{F}_2^n$. Note that for a given sparsity, the set of signals that are used for simulating FIPA is also used for simulating SIPA. 

Let CMP, SUB and ADD refer to macro-operations involving comparison ($\max$ or $\min$), subtraction and addition, respectively. Also, let $O_{\mathrm{CMP}}$, $O_{\mathrm{SUB}}$, and $O_{\mathrm{ADD}}$ refer to the number of CMP, SUB, and ADD operations, respectively, each having value $1$ (see \cite[page 10]{Fog}). 

For algorithm complexity calculation, we take into account the number of elementary operations needed for each CN to VN (and VN to CN) message lower bound computation. The complexity of upper bound calculation is ignored since by doing so, the comparison of complexities between $\mathrm{SIPA}(\mathbf{y,\mathbf{A}})$ and $\mathrm{FIPA}(\mathbf{y,\mathbf{A}})$ is unaffected. Since binary signals and measurement matrices are considered, the multiplication and division operations involved in message computation can be ignored. Then, according to (\ref{eq:mu_vc_fld}), the complexity of a single VN to CN message is given by $(\gamma-1)\gamma \times O_{\mathrm{CMP}}$. On the other hand, according to (\ref{eq:mu_cv_fld}), the complexity of a single CN to VN message is given by $O_{\mathrm{SUB}}+(\rho-1) \times O_{\mathrm{ADD}}$. Note that Step 11 of Algorithm \ref{alg:sipa} incurs an additional complexity of computing $p_{v_i}^{(\ell)}$ which is not present in FIPA. This added complexity is simply the number of VN to CN messages propagated by $\mathrm{SIPA}(\mathbf{y,\mathbf{A}})$ times $O_{\mathrm{ADD}}$. The final message passing complexity, in terms of the average number of elementary computations, is obtained by multiplying the complexities of the above macro-operations with the average number of CN to VN, or VN to CN messages executed during simulation for a fixed sparsity. We then add the resulting average complexities of both message types to get the final average complexity. For SIPA, we add to this result the average number of VN to CN messages propagated during simulation to take into account the cost of computing~$p_{v_i}^{(\ell)}$. 
 
\indent The PCR vs. sparsity plots for the $\mathbf{A}_1$ measurement matrix based on flooding and sequential interval passing techniques are shown in Fig. \ref{fig:res12}. For this measurement matrix, SIPA requires fewer iterations on average compared to FIPA for the same PCR given a fixed signal sparsity. At a sparsity of $0.06$, SIPA requires $42.4$\% fewer iterations on average than FIPA using the $A_1$ measurement matrix. Additionally, in Table \ref{tab:A3}, we compare the average number of messages propagated by the measurement matrix to attain the results shown in Fig. \ref{fig:res12}. The last row of the table shows the percentage reduction of complexity achieved by SIPA when compared to its flooding counterpart. These percentages are calculated by taking the ratio of the average complexities shown in rows five and six of the table. From the last row of Table \ref{tab:A3} we note that at a sparsity of $0.06$, SIPA requires $36$\% fewer arithmetic operations with respect to FIPA for computing all the messages. Moreover, the table reveals that regardless of the measurement matrix, the complexity advantage of sequential scheduling diminishes as the signal sparsity becomes too large. 

\vspace{-0.35cm}
\begin{figure}[h]
  \centering
  \includegraphics[scale=0.35]{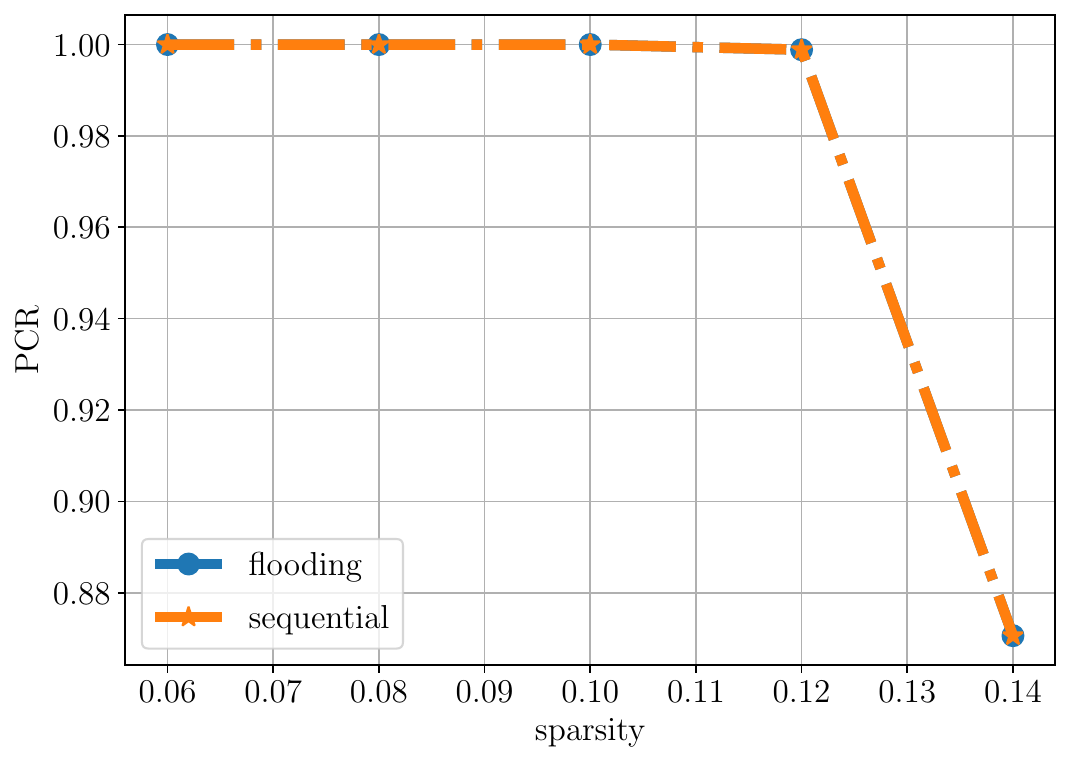}
  \vspace{-0.12cm}
  \caption{Reconstruction performance of $\mathbf{A}_1$ using different interval passing schemes.}
  \label{fig:res12}
\end{figure}

\vspace{-0.6cm}
{\linespread{0.2}\selectfont
\begin{table}[h]
\small
\centering
	\begin{tabular}{|c|c|c|c|}
		\hline
		\textbf{sparsity}
		&$\mathbf{0.06}$ 
		&$\mathbf{0.10}$    
		&$\mathbf{0.14}$ \\
		\hline\hline
		avg. \# of CN to VN msg. fld.&$8163$&$12700$&$38548$\\ 
		\hline
		avg. \# of CN to VN msg. seq.&$4703$&$7866$&$28636$\\ 
		\hline
		avg. \# of VN to CN msg. fld.&$8163$&$12700$&$38548$\\ 
		\hline
		avg. \# of VN to CN msg. seq.&$5000$&$8989$&$40144$\\ 
		\hline
		avg. complexity fld.&$106120$&$165102$&$501120$\\ 
		\hline
		avg. complexity seq.&$67920$&$117986$&$481458$\\ 
		\hline
		avg. \% reduction of complexity&$35.99$&$28.54$&$3.92$\\ 
		\hline
	\end{tabular}
	\quad
	\vspace{-0.15cm}
	\caption{Complexity results using the $\mathbf{A}_1$ measurement matrix for attaining the results shown in Fig. \ref{fig:res12}.}
\label{tab:A3} 
\end{table}
}

\vspace{-0.4cm}
\section{Conclusion}
In this paper, we improved the IPA via sequential scheduling of the CNs in the Tanner graph of an LDPC measurement matrix. SIPA propagates a significantly lower number of iterations compared to the standard FIPA, and by doing so the accuracy of sparse signal reconstruction is not compromised. We also show that for some LDPC measurement matrix, the FER of $\mathrm{SIPA}(\mathbf{y},\mathbf{A})$ is upper bounded by the FER of $\mathrm{FIPA}(\mathbf{y},\mathbf{A})$. Future research will take into account methods for optimizing $\mathcal{G}_{\mathbf{A}}$ to further reduce the sequential scheduling complexity of $\mathrm{SIPA}(\mathbf{y},\mathbf{A})$.

\newpage
\bibliographystyle{plain}
\bibliography{Bib}

\begin{thebibliography}{10}

\bibitem{Bayati}
Mohsen Bayati and Andrea Montanari.
\newblock The dynamics of message passing on dense graphs, with applications to
  compressed sensing.
\newblock {\em IEEE Trans. on Inf. Theory}, 57(2):764--785, 2011.

\bibitem{Candes}
E.~J. Candes, J.~K. Romberg, and T.~Tao.
\newblock Stable signal recovery from incomplete and inaccurate measurements.
\newblock {\em Communications on Pure and Applied Mathematics},
  59(8):1207--1223, 2006.

\bibitem{Chandar}
V.~Chandar, D.~Shah, and G.~W. Wornell.
\newblock A simple message-passing algorithm for compressed sensing.
\newblock {\em Proceedings of the IEEE Int'l Symp. on Inf. Theory}, pages
  1968--1972, Jun. 2010.

\bibitem{CDS98}
S.~S. Chen, D.~L. Donoho, and M.~A. Saunders.
\newblock Atomic decomposition by basis pursuit.
\newblock {\em SIAM J. Comput}, 20(1):33--61, Aug. 1998.

\bibitem{Don06}
D.~L. Donoho.
\newblock Compressed sensing.
\newblock {\em IEEE Trans. on Inf. Theory}, 52(4):1289--1306, Apr 2006.

\bibitem{DonMalMon09}
D.~L. Donoho, A.~Maleki, and A.~Montanari.
\newblock Message passing algorithms for compressed sensing: I. motivation and
  construction.
\newblock {\em Proc.~IEEE Inf. Theory Workshop (ITW), Cairo}, July 2010.

\bibitem{Fog}
Agner Fog.
\newblock Lists of instruction latencies, throughputs and micro-operation
  breakdowns for {I}ntel, {AMD}, and {VIA} {CPU}s.
\newblock [Online]. Available:
  https://www.agner.org/optimize/instruction\_tables.pdf, 2022.

\bibitem{Chunli}
Chunli Guo and Mike~E. Davies.
\newblock Near optimal compressed sensing without priors: Parametric sure
  approximate message passing.
\newblock {\em IEEE Trans. on Signal Processing}, 63(8):2130--2141, 2015.

\bibitem{reldec}
Salman Habib, Allison Beemer, and Jörg Kliewer.
\newblock {RELDEC}: Reinforcement learning-based decoding of moderate length
  {LDPC} codes.
\newblock {\em IEEE Trans. on Commun.}, 71(10):5661--5674, 2023.

\bibitem{mycs}
Salman Habib and Jörg Kliewer.
\newblock Algebraic optimization of binary spatially coupled measurement
  matrices for interval passing.
\newblock In {\em 2018 IEEE Inf. Theory Workshop (ITW)}, pages 1--5, 2018.

\bibitem{Shantharam}
Shantharam Kalipatnapu and Indrajit Chakrabarti.
\newblock Low-complexity interval passing algorithm and {VLSI} architecture for
  binary compressed sensing.
\newblock {\em IEEE Trans. Very Large Scale Integr. (VLSI) Syst.},
  28(5):1283--1291, 2020.

\bibitem{Kim15}
Taejoon Kim and David~J. Love.
\newblock {Virtual AoA and AoD estimation for sparse millimeter wave MIMO
  channels}.
\newblock In {\em 2015 IEEE 16th Intl. Workshop on Signal Processing Advances
  in Wireless Commun. (SPAWC)}, pages 146--150, 2015.

\bibitem{Krishnan}
Anantha~Raman Krishnan, Swaminathan Sankararaman, and Bane Vasic.
\newblock Graph-based iterative reconstruction of sparse signals for compressed
  sensing.
\newblock In {\em 2011 10th International Conference on Telecommunication in
  Modern Satellite Cable and Broadcasting Services (TELSIKS)}, volume~1, pages
  133--137, 2011.

\bibitem{10206981}
Dang~Qua Nguyen and Taejoon Kim.
\newblock On the stability of approximate message passing with independent
  measurement ensembles.
\newblock In {\em 2023 IEEE International Symposium on Information Theory
  (ISIT)}, pages 779--784, 2023.

\bibitem{Hoa}
Hoa~V. Pham, Wei Dai, and Olgica Milenkovic.
\newblock Sublinear compressive sensing reconstruction via belief propagation
  decoding.
\newblock In {\em 2009 Proc. of the IEEE Int'l Symp. on Inf. Theory}, pages
  674--678, 2009.

\bibitem{Ravanmehr}
V.~Ravanmehr, L.~Danjean, B.~Vasic, and D.~Declercq.
\newblock Interval-passing algorithm for non-negative measurement matrices:
  Performance and reconstruction analysis.
\newblock {\em IEEE J. Sel. Topics Circuits and Systems}, 2(3):424--432, Sep.
  2012.

\bibitem{Sarvotham}
Shriram Sarvotham, Dror Baron, and Richard Baraniuk.
\newblock Compressed sensing reconstruction via belief propagation.
\newblock 2006.

\bibitem{Nikolajs}
Nikolajs Skuratovs and Michael~E. Davies.
\newblock Compressed sensing with upscaled vector approximate message passing.
\newblock {\em IEEE Trans. on Inf. Theory}, 68(7):4818--4836, 2022.

\bibitem{Tan81}
R.~M. Tanner.
\newblock A recursive approach to low complexity codes.
\newblock {\em IEEE Trans. on Inf. Theory}, 27(5):547--553, Sep 1981.

\end{thebibliography}
\end{document}